\documentclass[runningheads]{llncs} 
\usepackage[utf8]{inputenc}
\usepackage{amsmath}
\usepackage{dsfont}
\usepackage{paralist}
\usepackage{graphicx}
\usepackage{amssymb}
\usepackage{comment}
\usepackage[colorlinks=true,urlcolor=blue,citecolor=red,linkcolor=black]{hyperref}
\usepackage{todonotes}
\newcommand{\R}{\ensuremath{\mathds R}}

\newcommand{\nice}{philogeodetic}

\graphicspath{{figures/}}
\usepackage{color,xcolor,graphicx,colortbl}
\usepackage[compatibility=false,font=small]{caption}
\usepackage[font=small]{subcaption}

\title{Drawing Shortest Paths in Geodetic Graphs\thanks{This research began at the Graph and Network Visualization Workshop 2019 (GNV'19) in Heiligkreuztal.
S.~C.~is funded by the German Research Foundation DFG – Project-ID 50974019 – TRR 161 (B06). M.~H.~is supported by the Swiss National Science Foundation within the collaborative DACH project \emph{Arrangements and Drawings} as SNSF Project 200021E-171681. S.~K.~is supported by NSF grants CCF-1740858, CCF-1712119, and
DMS-1839274.}}
\author{Sabine Cornelsen\inst{1}\orcidID{0000-0002-1688-394X}
  \and
  Maximilian~Pfister\inst{2}\orcidID{0000-0002-7203-0669}
  \and
  Henry~F\"orster\inst{2}\orcidID{0000-0002-1441-4189}
  \and
  Martin~Gronemann\inst{3}\orcidID{0000-0003-2565-090X}
  \and 
  Michael~Hoffmann\inst{4}\orcidID{0000-0001-5307-7106} 
  \and
  Stephen~Kobourov\inst{5}\orcidID{0000-0002-0477-2724}
  \and
  Thomas~Schneck\inst{2}\orcidID{0000-0003-4061-8844}
  }

\authorrunning{S.~Cornelsen et al.}

\institute{University of Konstanz, Germany 
\email{sabine.cornelsen@uni-konstanz.de}
\and
University of T\"ubingen, Germany
\email{\{pfister,foersth,schneck\}@informatik.uni-tuebingen.de}\and
University of Osnabr\"uck, Germany\\
\email{martin.gronemann@uni-osnabrueck.de}\and
Department of Computer Science, ETH Z{\"u}rich, Switzerland\\
 \email{hoffmann@inf.ethz.ch}
 \and
Department of Computer Science, University of Arizona, USA
\email{kobourov@cs.arizona.edu}
}

\begin{document}

\maketitle

\begin{abstract}
  Motivated by the fact that 
  in a space where shortest paths are unique, 
  no two shortest paths meet twice, %~--~provided that shortest paths are unique~--~
  we study a question posed by Greg Bodwin: Given a geodetic graph
  $G$, i.e., an unweighted graph in which the shortest path between
  any pair of vertices is unique, is there a \emph{\nice{}} drawing
  of $G$, i.e., a drawing of $G$ in which the curves of any two
  shortest paths meet at most once?
  We answer this question in the negative by showing the existence of geodetic graphs that require some pair of shortest paths
  to cross at least four times. The bound on the number of crossings is tight for the class of graphs we construct. Furthermore, we exhibit geodetic graphs of diameter two that do not admit a \nice{} drawing.

  \end{abstract}

\section{Introduction}
Greg Bodwin~\cite{bodwin:soda19} examined the structure of 
%the set of 
shortest paths in graphs with edge weights that guarantee that the shortest path between any pair of vertices is unique. Motivated by the fact that a set of unique shortest paths is \emph{consistent} in the sense that no two such paths can ``intersect, split apart, and then intersect again",  he conjectured that 
%the structure of shortest paths in graphs is ``nice". Specifically, he conjectured \todo{MH: Was that not for geodetic graphs only?} that
if the shortest path between any pair of vertices in a graph is unique then the graph can be drawn so that any two shortest paths meet at most once.
Formally, a \emph{meet} 
of two Jordan curves $\gamma_1,\gamma_2:[0,1]\to\R^2$ is a pair of maximal intervals $I_1, I_2\subseteq[0,1]$ for which there is a bijection $\iota: I_1 \rightarrow I_2$ so that %$I\cap\{0,1\}=\emptyset$ and 
$\gamma_1(x)=\gamma_2(\iota(x))$ for all $x \in I_1$. A \emph{crossing} is a meet with $(I_1\cup I_2)\cap\{0,1\}=\emptyset$. %\todo{SC: Do we need this definition of crossing anywhere? MH: We talk about crossings later, when using the Crossing Lemma.}
%Formally, a \emph{meet} 
%of two curves $\gamma_1,\gamma_2:[0,1]\to\R^2$ is a maximal interval $I\subseteq[0,1]$ so that %$I\cap\{0,1\}=\emptyset$ and 
%$\gamma_1(x)=\gamma_2(x)$, for all $x\in I$. A \emph{crossing} is a meet with $I\cap\{0,1\}=\emptyset$.
%\todo{SC: If the definition of crossing should be consistent with the consistency definition of unique shortest paths in the Bodwin paper then $I\cap\{0,1\}\neq\emptyset$ should also be allowed, true?}
%\todo{MH: I changed the term to ``meet''.}
Two curves %\emph{meet} if they have a meet, and they 
\emph{meet $k$ times}
% for $k\in\N$,
if they have $k$ pairwise distinct meets. 
For example, shortest paths in a simple polygon (geodesic paths) have the property that they meet at most once~\cite{lp-esppr-84}.

%\todo[inline]{SC: 1. Do we need this formal definition of crossing in the introduction at all? We never refer to it later.\\ 
%2. Why calling the curves first $\gamma$ and later $\varphi$?\\
%3. The way crossings are defined now also includes touchings, i.e. a curve approaches another curve from one side and leaves it to the same side. Do we really want that? \\
%4. We later only count edge crossings whereas this definition also counts crossings at common inner vertices of both paths. Should we say something about that?
%}

\begin{figure}[t]
        \centering
        \includegraphics{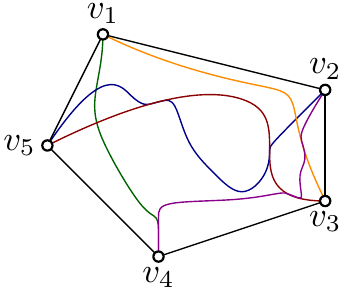}
        \caption{A drawing of the geodetic graph $K_5$. It has a crossing formed by edges $v_1v_3$ and $v_2v_5$. In addition, edges $v_1v_4$ and $v_2v_4$ meet but do not cross since their meet includes vertex $v_4$. Finally, edges $v_2v_5$ and $v_3v_5$ meet twice violating the property of \nice~drawings.}
        \label{fig:philo-example}
    \end{figure}

A \emph{drawing} of a graph $G$ in $\R^2$ maps the vertices to pairwise distinct points
and maps each edge to a Jordan arc between 
%(the images of) 
the two end-vertices that is disjoint from 
%(the image of) 
any other vertex. Drawings extend in a natural fashion to paths: Let $\varphi$ be a drawing of $G$, and let $P=v_1,\ldots,v_n$ be a path in $G$. Then let $\varphi(P)$ denote the Jordan arc that is obtained as the composition of the curves  $\varphi(v_1v_2),\ldots,\varphi(v_{n-1}v_n)$. A drawing $\varphi$ of a graph $G$ is \emph{\nice} 
%\todo{SC: Henry suggested "topologically consistent" because it reuses the word consistent from the Bodwin paper. However, this does not tell much about the desired properties. Any opinions?\\
%MH: geodetically consistent? or philogeodetic?} 
if for every pair $P_1,P_2$ of shortest paths in $G$ the curves $\varphi(P_1)$ and $\varphi(P_2)$ meet at most once.  

 An unweighted graph is \emph{geodetic} if there is a unique shortest path between every pair of vertices.
 %when all edge weights are 1.\todo{MH: This sounds as if unit weights were required. Just call it an unweighted graph?} 
 Trivial examples of geodetic graphs are trees 
and complete graphs. Observe that any two shortest paths in a geodetic graph are either disjoint or they intersect in a path.
%share \todo{SC: This is now not intersection of curves but of paths in the abstract graph, so it is not meet} exactly one connected component. 
Thus, a planar drawing of a planar geodetic graph is \nice. Also every straight-line drawing of a complete graph is \nice. Refer to \figurename~\ref{fig:philo-example} for an illustration of a drawing of a complete graph that is not \nice; this example also highlights some of the concepts discussed above. 
%\todo{SC: I'm not sure whether the example in Fig.~\ref{fig:philo-example} is helpful. Anyway the first two "problematic" properties stated in the caption of Fig.~\ref{fig:philo-example} are not forbidden for a \nice{} drawing. Actually, I'm not sure, whether an example is reasonable at all. Anyway, we'll have a negative example later in Fig.~\ref{fig:k82}. A positive example, on the other hand, does not illustrate much. }\todo{MH: I agree.}  
%Similarly, when studying the structure of shortest paths in graphs, \todo{MH: This statement sounds weird, does it refer to shortest path trees? If ``crossing'' is only about edge crossings in a drawing, what does it mean for an abstract graph? If it also includes common vertices, it seems wrong in general if the graph is not geodetic.} pairs of such paths never cross more than once. 
It is a natural question to ask whether every (geodetic) graph admits a philogeodetic drawing.

\paragraph{Results.}
%We consider the class of {\em geodetic graphs}, that is graphs in which there is a unique %shortest path between every pair of vertices. 
We show that there exist geodetic graphs that require some pair of shortest paths to meet at least four times (Theorem~\ref{thm:1}). 
%% MH: removed the following sentence (because this is how we defined drawing, anyway).
%% This is even true
%both in a geometric (straight-line) drawing but also 
% in any topological drawing. 
The idea is to start with a sufficiently large complete graph and subdivide every edge exactly twice.
%and count the crossings in the resulting graph. 
The Crossing Lemma~\cite{pt-wcn-00} can be used to show that some pair of shortest paths must cross at least four times. By increasing the number of subdivisions per edge, we can reduce the density and obtain sparse counterexamples. The bound on the number of crossings is tight because 
%it is not hard to show that 
any uniformly subdivided $K_n$ can be drawn so that every pair of shortest paths meets at most four times (Theorem~\ref{thm:unifour}).
%Note that this example can be generalized by subdividing edges further and thus reduce the density of the counterexample. 

On one hand, our construction yields counterexamples of diameter five.
%the diameter of the counterexamples with two subdivisions per edge is five. %at least 5 and increases with an increasing number of subdivision vertices.
%\todo{HF: This should only be true for edges subdivided twice.}.
On the other hand, the unique graph of diameter one is the complete graph, which is geodetic and admits a \nice{} drawing % so that every pair of shortest paths crosses at most once 
(e.g., any straight-line drawing since all unique shortest paths are single edges). Hence, it is natural to ask what is the largest %integer 
$d$ so that every geodetic graph of diameter $d$ admits a \nice{} drawing. We show that $d=1$ by exhibiting an infinite family of geodetic graphs of diameter two that do not admit \nice{} drawings (Theorem~\ref{thm:d2}). The construction is based on incidence graphs of finite affine planes.  
%We then consider diameter-2 geodetic graphs and show that there exist such low diameter counterexamples. The graph used comes from a sufficiently large affine plane and 
The proof also relies on the crossing lemma.
%A more sophisticated counter argument is needed here.

\paragraph{Geodetic graphs.}
Geodetic graphs were introduced by Ore who asked for a characterization as Problem~3 in Chapter~6 of his book ``Theory of Graphs''~\cite[p.~104]{o-tg-62}. An asterisk flags this problem as a research question, which seems justified, as more than sixty years later a full characterization is still elusive. 

Stemple and Watkins~\cite{sw-pgg-68,w-cpgg-67} and
Plesn{\'i}k~\cite{p-tcgg-77} resolved the planar case by showing that
a connected planar graph is geodetic if and only if every
block is (1)~a single edge, (2)~an odd cycle, or
% (3)~a geodetic subdivision of $K_4$. A subdivision $G'$ of $G\simeq K_4$ is geodetic if (a)~for any pair of vertices $u,v$ in $G$, the path in $G'$ that corresponds to the edge $uv$ in $G$ is a (hence unique) shortest path in $G'$; (b)~any $k$-cycle of $G$, for $k\in\{3,4\}$, is subdivided an even number of times in $G'$; and (c)~every $4$-cycle of $G$ is subdivided the same number of times in $G'$.
(3)~stems from a $K_4$ by iteratively choosing a vertex $v$ of the
$K_4$ and subdividing the edges incident to $v$ uniformly.
Geodetic graphs of diameter two were fully characterized by Scapellato~\cite{s-ggdtsrs-86}. They include the Moore graphs~\cite{hs-mgd-60} and graphs constructed from a generalization of affine planes.
Further constructions for geodetic graphs were given by Plesn\'ik~\cite{p-tcgg-77,plesnik84}, Parthasarathy and Srinvasan~\cite{parthasarathySrinivasan82}, and Frasser and Vostrov~\cite{frasserVostrov-arxiv16}.

Plesn\'ik~\cite{p-tcgg-77} and Stemple \cite{stemple79} proved that a
geodetic graph is homeomorphic to a complete graph if and only if it
is obtained from a complete graph $K_n$ by iteratively choosing a
vertex $v$ of the $K_n$ and subdividing the edges incident to $v$
uniformly.
%This includes the graphs obtained from a complete graph by
%uniformly subdividing each edge an even number of times.
A graph is geodetic if it is obtained from any geodetic graph by
uniformly subdividing each edge an even number of
times~\cite{parthasarathySrinivasan82,plesnik84}. 
% New explanation for odd subdivisions start
However, the graph $G$ obtained by uniformly subdividing each edge of a complete  graph $K_n$ an odd number of times is not geodetic: Let $u,v,w$ be three vertices of $K_n$ and let $x$ be the middle subdivision vertex of the edge $uv$. Then there are two shortest $x$-$w$-paths in $G$, one containing $v$ and one containing $u$. 

\section{Subdivision of a Complete Graph}

The complete graph $K_n$ is geodetic and rather dense. However, all shortest paths are very short, as they comprise a single edge only.
So despite the large number of edge crossings in any drawing, every pair of shortest paths meets at most once, as witnessed, for instance, by any straight-line drawing of $K_n$. In order to lengthen the shortest paths it is natural to consider subdivisions of $K_n$. 

As a first attempt, one may want to ``take out'' some edge $uv$ by subdividing it many times. However, Stemple~\cite{stemple79} has shown that in a geodetic graph every path where all internal vertices have degree two must be a shortest path. Thus, it is impossible to take out an edge using subdivisions.
So we use a different approach instead, where all edges are subdivided uniformly.

\begin{theorem}\label{thm:1}
    There exists an infinite family of sparse geodetic graphs for which in any %(topological) 
    drawing in $\R^2$ some pair of shortest paths meets at least four times.
\end{theorem}
\begin{proof}
  Take an even number $t$ and a complete graph $K_s$ for some $s \in \mathds N$. Subdivide each edge 
  $t$ times.
  % By Theorem~\ref{thm:subdivision}
  The resulting graph $K(s,t)$ is geodetic. See \figurename~\ref{fig:k82} for a drawing of $K(8,2)$. Note that $K(s,t)$ has $n=s+t\binom{s}{2}$ vertices and $m=(t+1)\binom{s}{2}$ edges, with $m\in O(n)$, for $s$ fixed and $t$ sufficiently large. Consider a
  %n arbitrary 
  drawing $\Gamma$ of $K(s,t)$. 

    Let $B$ denote the set of $s$ \emph{branch vertices} in $K(s,t)$, 
    %i.e., the vertices of degree $>2$, 
    which correspond to the vertices of the original $K_s$. For two distinct vertices $u,v\in B$, let $[uv]$ denote the shortest $uv$-path in $K(s,t)$, which corresponds to the subdivided edge $uv$ of the underlying $K_s$. As $t$ is even, the path $[uv]$ consists of $t+1$ (an odd number of) edges. For every such path $[uv]$, with $u,v\in B$, we charge the crossings in $\Gamma$ along the $t+1$ edges of $[uv]$ to one or both of $u$ and $v$ as detailed below; see \figurename~\ref{fig:charging} for illustration.
    \begin{itemize}
        \item Crossings along an edge that is closer to $u$ than to $v$ are charged to $u$; 
        \item crossings along an edge that is closer to $v$ than to $u$ are charged to $v$; \emph{and}
        \item crossings along the single central edge of $[uv]$ are charged to both $u$ and~$v$. 
    \end{itemize}
    
    \begin{figure}[htbp]
        \centering
        \includegraphics{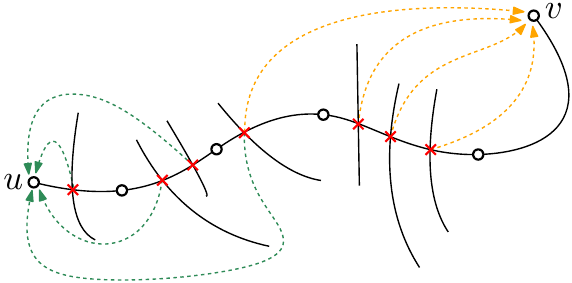}
        \caption{Every crossing is charged to at least one endpoint of each of the two involved (independent) edges. Vertices are shown as white disks, crossings as red crosses, and charges by dotted arrows. The figure shows an edge $uv$ that is subdivided four times, splitting it into a path with five segments. A crossing along any such segment is assigned to the closest of $u$ or $v$. For the central segment, both $u$ and $v$ are at the same distance, and any crossing there is assigned to both $u$ and $v$.}
        \label{fig:charging}
    \end{figure}

        Let $\Gamma_s$ be the drawing of $K_s$ %\todo{SC: In what sense is $K_s$ a subset of $K(s,t)$?}
    % \subset K(s,t)$
    induced by $\Gamma$: every vertex of $K_s$ is placed at the position of the corresponding branch vertex of $K(s,t)$ in $\Gamma$ and every edge of $K_s$ is drawn as a Jordan arc along the corresponding path of $K(s,t)$ in $\Gamma$. Assuming $\binom{s}{2}\ge 4s$ (i.e., $s\ge 9$), by the Crossing Lemma~\cite{pt-wcn-00}, at least 
    \[
    \frac{1}{64}\frac{{\binom{s}{2}}^3}{s^2}=\frac{1}{512}s(s-1)^3\ge c\cdot s^4
    \]
    pairs of independent edges cross in $\Gamma_s$, for some constant $c$. Every crossing in $\Gamma_s$ corresponds to a crossing in $\Gamma$ and is charged to at least two (and up to four) vertices of $B$. Thus, the overall charge is at least $2cs^4$, and at least one vertex $u\in B$ gets at least the average charge of $2cs^3$. 
    
    Each charge unit 
    %unit of this charge 
    corresponds to a crossing of two independent edges in $\Gamma_s$, which is also charged to at least one other vertex of $B$. Hence, there is a vertex $v\ne u$ so that at least $2cs^2$ crossings are charged to both $u$ and $v$. Note that there are only $s-1$ edges incident to each of $u$ and $v$, and the common edge 
    %\todo{MH: Rev3 complains ``you mention the common edge uv, but there is no such edge any more. Please reformulate so that it is clear what you mean.`` But we are talking about a drawing $\Gamma_s$ of $K_s$, so I think the term common edge is clear.}
    $uv$ is not involved in any of the charged crossings (as adjacent rather than independent edge). Let $E_x$, for $x\in B$, denote the set of edges of $K_s$
    % \subset K(s,t)$
    that are incident to $x$.
    
    We claim that there are two pairs of mutually crossing edges incident to $u$ and $v$, respectively; that is, there are sets $C_u\subset E_u\setminus\{uv\}$ and $C_v\subset E_v\setminus\{uv\}$ with $|C_u|=|C_v|=2$ so that $e_1$ crosses $e_2$, for all $e_1\in C_u$ and $e_2\in C_v$.
    
    Before proving this claim, we argue that establishing it completes the proof of the theorem. By our charging scheme, every crossing $e_1\cap e_2$ happens at an edge of the path $[e_1]$ in $\Gamma$ that is at least as close to $u$ as to the other endpoint of $e_1$. Denote the three vertices that span the edges of $C_u$ by $u,x,y$. Consider the two subdivision vertices $x'$ along $[ux]$ and $y'$ along $[uy]$ that form the endpoint of the middle edge closer to $x$ and $y$, respectively, than to $u$; see \figurename~\ref{fig:adjacent} for illustration. 
    
\begin{figure}[htbp]
    \centering
    \includegraphics{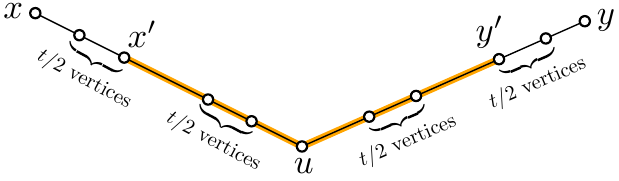}
    \caption{Two adjacent edges $ux$ and $uy$, both subdivided $t$ times, and the shortest path between the ``far'' endpoints $x'$ and $y'$ of the central segments of $[ux]$ and $[uy]$.} %shown in bold and orange.}
    \label{fig:adjacent}
\end{figure}  
    
    The triangle $uxy$ in $K_s$ corresponds to an odd cycle of length $3(t+1)$ in $K(s,t)$. So the shortest path between $x'$ and $y'$ in $K(s,t)$ has length $2(1+t/2)=t+2$ and passes through $u$, whereas the path from $x'$ via $x$ and $y$ to $y'$ has length $3(t+1)-(t+2)=2t+1$, which is strictly larger than $t+2$ for $t\ge 2$. It follows that the shortest path between $x'$ and $y'$ in $K(s,t)$ is crossed by both edges in $C_v$. A symmetric argument yields two subdivision vertices $a'$ and $b'$ along the two edges in $C_v$ so that the shortest $a'b'$-path in $K(s,t)$ is crossed by both edges in $C_u$. By definition of our charging scheme (that charges only ``nearby'' crossings to a vertex), the shortest paths $x'y'$ and $a'b'$ in $K(s,t)$ have at least four crossings.
  
    It remains to prove the claim. To this end, consider the bipartite graph $X$ on the vertex set $E_u\cup E_v$ where two vertices are connected if the corresponding edges are independent and cross in $\Gamma_s$. Observe that two sets $C_u$ and $C_v$ of mutually crossing pairs of edges (as in the claim) correspond to a $4$-cycle $C_4$ in $X$. So suppose for the sake of a contradiction that $X$ does not contain $C_4$ as a subgraph. Then by the K{\H o}v{\'a}ri-S{\'o}s-Tur{\'a}n Theorem~\cite{kst-pkz-54} the graph $X$ has $O(s^{3/2})$ edges. But we already know that $X$ has at least $2cs^2=\Omega(s^2)$ edges, which yields a contradiction. Hence, $X$ is not $C_4$-free and the claim holds.\qed
\end{proof}

The bound on the number of crossings in Theorem~\ref{thm:1} is tight.

\begin{theorem}\label{thm:unifour}
 A graph obtained from a complete graph by subdividing the edges uniformly an even number of times
  %Any uniformly subdivided (an even number of times) $K_n$ 
  can be drawn so that every pair of shortest paths crosses at most four times.
\end{theorem}

\begin{proof}[Sketch]
%We only sketch the construction, a proof of Theorem~\ref{thm:unifour} can be found in \cite{arxiv-version}.
\label{PAGE:construction}Place the vertices in convex position. Draw the subdivided edges along straight-line segments. For each edge, put half of the subdivision vertices very close to one endpoint and the other half very close to the other endpoint (\figurename~\ref{fig:k82}). As a result, all crossings  fall into the central segment of the path. %\figurename~\ref{fig:k82} shows a corresponding drawing for $K(8,2)$ as an example.
\qed
\end{proof}

\begin{figure}[htbp]
    \centering
    \includegraphics[scale=.78]{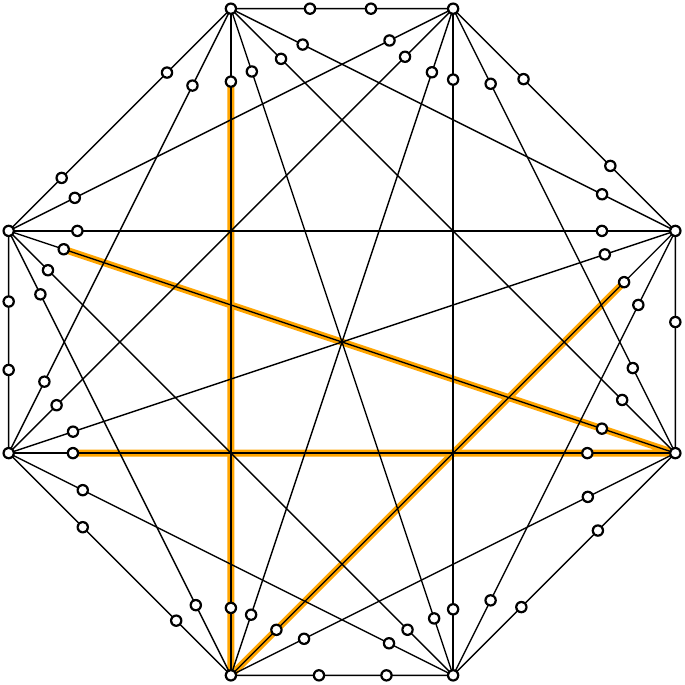}
    \caption{A drawing of $K(8,2)$, the complete graph $K_8$ where every edge is subdivided twice, so that every pair of shortest paths meets at most four times. Two shortest paths that meet four times are shown bold and orange.}
    \label{fig:k82}
\end{figure}

%\section{Two Counter Examples}
%
%In this section we give two geodetic graphs that cannot be drawn in
%the plane such that any two shortest paths cross at most once.

\section{Graphs of Diameter Two}
\label{sec:diam2}
In this section we give examples of geodetic graphs of diameter two
that cannot be drawn in the plane such that any two shortest paths
meet at most once.

An \emph{affine plane} of order $k \geq 2$ consists of a set of lines
and a set of points with a containment relationship such that (i) each
line contains $k$ points, (ii) for any two points there is a unique
line containing both, (iii) there are three points that are not
contained in the same line, and (iv) for any line $\ell$ and any point
$p$ not on $\ell$ there is a line $\ell'$ that contains $p$, but no
point from $\ell$. Two lines that do not contain a common point are
\emph{parallel}. Observe that each point is contained in $k+1$
lines. Moreover, there are $k^2$ points and $k+1$ classes of parallel
lines each containing $k$ lines. The 2-dimensional vector space
$\mathds F^2$ over a finite field $\mathds F$ of order $k$ with the
lines $\{(x,mx+b);\;x\in\mathds F\}$, $m,b \in \mathds F$ and
$\{(x_0,y);\; y \in \mathds F\}$, $x_0 \in \mathds F$ is a finite
affine plane of order $k$. Thus, there exists a finite affine plane of
order $k$ for any $k$ that is a prime power (see, e.g., \cite{hp-73-pp}). 

\label{def:gk}Scapellato~\cite{s-ggdtsrs-86} showed how %an infinite family of
to construct geodetic graphs of diameter two as follows: Take a
finite affine plane of order $k$.  Let $L$ be the set of lines and let
$P$ be the set of points of the affine plane. Consider now the graph
$G_k$ with vertex set $L \cup P$ and the following two types of edges:
There is an edge between two lines if and only if they are
parallel. There is an edge between a point and a line if and only if
the point lies on the line; see Fig.~\ref{fig:gk}. There are no edges between points. 
%In Appendix~\ref{app:gk}, we prove 
It is easy to check that $G_k$ is a geodetic graph of diameter two.

\begin{figure}[htbp]
    \centering
    \includegraphics{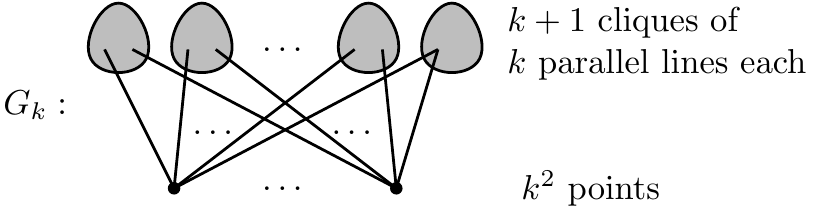}
    \caption{Structure of the graph $G_k$.}
    \label{fig:gk}
\end{figure}

\begin{theorem}\label{thm:d2}
  There are geodetic graphs of diameter two that cannot be drawn
  in the plane such that any two shortest paths meet at most once.
\end{theorem}
\begin{proof}
  Let $k \geq 129$ be such that there exists an affine plane of order
  $k$ (e.g., the prime $k=131$). Assume there was a drawing of $G_k$ in which
  any two shortest paths meet at most once. Let $G$ be the bipartite
  subgraph of $G_k$ without edges between lines. Observe that any path
  of length two in $G$ is a shortest path in $G_k$.  As $G$ has
  $n=2k^2 + k$ vertices and $m=k^2(k+1) > kn/2$ edges, we have
  $m>4n$, for $k\geq 8$.  Therefore, by the Crossing Lemma~\cite[Remark~2 on
  p.~238]{pt-wcn-00} there are at least $m^3/64n^2 > k^3 n/512$
  crossings between independent edges in~$G$.

  Hence, there is a vertex $v$ such that the edges incident to
  $v$ are crossed more than $k^3/128$ times by edges not incident to
  $v$. By assumption, (a)~any two edges meet at most once, (b)~any edge meets any pair of adjacent edges at most once, and (c)~any pair of adjacent edges meets any  pair of adjacent edges at most once.
  %%no two edges incident to $v$ are crossed by the same edge, (c) no edge incident to $v$ is crossed by a pair of edges incident to the  same vertex, and (d) no two edges incident to $v$ can be crossed by a pair of edges incident to the same vertex. 
  Thus, the crossings
  with the edges incident to $v$ stem from a matching. It follows
  that there are at most $(n-1)/2 = (2k^2 + k -1)/2$ such
  crossings. However, $(2k^2 + k -1)/2 < k^3/128$, for $k \geq 129$.\qed
\end{proof}

\section{Open Problems}

We conclude with two open problems: (1)~Are there diameter-2 geodetic graphs with edge density $1+\varepsilon$ %for arbitrary $\varepsilon > 0$ 
that do not admit a \nice{} drawing? (2)~What is the complexity of deciding if a geodetic graph admits a \nice{} drawing?

\bibliographystyle{splncs04}
\bibliography{references}

\clearpage

\begin{appendix}

\section{Proof of Theorem~\ref{thm:unifour}}
\begin{proof}
    Draw the graph as described on Page~\pageref{PAGE:construction} and as illustrated in Fig.~\ref{fig:k82} for $K(8,2)$. There are two different types of vertices, and six different types of shortest paths. Let $B$ denote the set of branch vertices, and let $S$ denote the set of subdivision vertices. Note that for every edge $uv$ of $K_n$, only the central segment of the subdivided path $[uv]$ may have crossings in the drawing. We claim that every shortest path in the graph contains at most two central segments in the drawing, from which the theorem follows immediately. Consider a pair $u,v$ of vertices.
    
    \smallskip\emph{Case~1:} $\{u,v\}\cap B\ne\emptyset$. Suppose without loss of generality that $u\in B$. If $v\in B$ or $v\in S$ subdivides an edge incident to $u$, then the shortest $uv$-path contains at most one central segment. Otherwise, $v\in S$ subdivides an edge $xy$ disjoint from $u$. One of $x$ or $y$, without loss of generality $x$ is closer to $v$. Then the shortest $uv$-path is $[vx][xu]$, which contains exactly one central segment, $[xu]$.
    
    \smallskip\emph{Case~2:} $u,v\in S$. If $u$ and $v$ subdivide the same edge, then the shortest $uv$-path contains at most one central segment. If $u$ and $v$ subdivide distinct adjacent segments, $xy$ and $xz$,  %respectively, 
    then the shortest $uv$-path is either $[ux][xv]$, which contains at most two central segments. Or the sum of the length of $[uy]$ and $[zv]$ is at most half of the number of subdivision vertices per edge and the shortest $uv$-path is $[uy][yz][zv]$, which then contains at most one central segment.  Otherwise, $u$ and $v$ subdivide disjoint segments, $xy$ and $wz$, %respectively, 
    where without loss of generality $x$ is closer to $u$ than $y$ and $w$ is closer to $v$ than $z$. Then the shortest $uv$-path is $[ux][xw][wv]$, which contains exactly one central segment, $[xw]$.\qed
  \end{proof}

  \section{Proof that $G_k$ (as Defined in Section~\ref{sec:diam2}) is Geodetic}
  \label{app:gk}

The following statement follows from Scapellato's classification~\cite{s-ggdtsrs-86}. As we need much less than this classification in its full generality, we provide an easy proof of what we use, for the sake of self-containment.

\begin{lemma}\label{lem:1}
  $G_k$ is a geodetic graph of diameter two.
\end{lemma}

\begin{proof}
  Two lines have distance one if they are parallel. Otherwise they
  share exactly one vertex and, hence, are connected by exactly one
  path of length two. For any two points there is exactly one line
  that contains both. Given a line $\ell$ and a point $p$ then either
  $p$ lies on $\ell$ and, thus, $p$ and $\ell$ have distance one. Or
  there is exactly one line $\ell'$ containing $p$ that is parallel to
  $\ell$ and, thus, there is exactly one path of length two between
  $\ell$ and $p$.\qed
\end{proof}

\end{appendix}

\end{document}